\documentclass[a4paper,10pt,reqno]{amsart}
\usepackage{amsmath,amssymb,amsthm}
\usepackage{amsaddr}
\usepackage{amsfonts}
\usepackage{multirow,bigdelim}
\usepackage{diagbox}
\usepackage{comment}
\usepackage{enumerate}
\usepackage{url}
\usepackage[lined,linesnumbered,boxed]{algorithm2e}
\usepackage{graphicx}
\usepackage[labelformat=simple]{subcaption}
\newcommand{\order}{{\rm O}}

\newcommand{\OPT}{{\rm OPT}}

\newcommand{\RSC}{{\rm RSC}}

\newcommand{\lca}{{\rm lca}}
\newcommand{\SC}{{\rm SC}}

\newcommand{\real}{{\mathbb R}}
\newcommand{\zahl}{{\mathbb Z}}
\newcommand{\natu}{{\mathbb N}}

\SetKwInput{KwIn}{Input}
\SetKwInput{KwOut}{Output}
\SetKw{KwBreak}{break}
\SetKwRepeat{Do}{do}{while}
\SetAlCapSkip{0.5em}
\SetAlCapSty{}
\theoremstyle{plain}
\newtheorem{theorem}{Theorem}[section]
\newtheorem{lemma}[theorem]{Lemma}
\newtheorem{proposition}[theorem]{Proposition}
\newtheorem{corollary}[theorem]{Corollary}

\theoremstyle{remark}

\newtheorem{defe}[theorem]{Definition} 
\allowdisplaybreaks[1]
\begin{document}
\title[Characterizing Admissible Objective Functions]{Characterizing Admissible Objective Functions for Hierarchical Clustering}
\author{Kazutoshi Ando$^*$}
\address{Faculty of Engineering, Shizuoka University,               Johoku 3-5-1, 432-8561  
		Hamamatsu, Japan
%              Tel.: +81-53-478-1263\\
%              Fax: +81-53-478-1263\\
}
\email{ando.kazutoshi@shizuoka.ac.jp}
\thanks{$^*$Corresponding author: \texttt{ando.kazutoshi@shizuoka.ac.jp}}
%\thanks{Supported by JSPS KAKENHI Grant Numbers 18K11180 and 22K11921.}
\author{Ryuki Tsukuba}
\address{Graduate School of Integrated Science and Technology, 
Shizuoka University, Hamamatsu, Shizuoka 432-8561, Japan}
\email{tsukuba.ryuki.19@shizuoka.ac.jp}
\date{June 15, 2026}
% 54E15, 05C65(hypergraphs),91C20(clustering), 62H30(classification), 51K05 (distance geometry - generaly theory)
\begin{abstract}
Hierarchical clustering is a fundamental task in data analysis, but classical methods have long lacked a principled objective function. Dasgupta~[STOC~2016] took an important step toward addressing this gap by proposing a well-motivated objective function for cluster trees. 
Cohen-Addad et al.~[J. ACM 2019] subsequently introduced the notion of admissibility: an objective function is admissible if, whenever the input similarity matrix admits generating trees, its minimizers are precisely those generating trees.
They also gave a necessary and sufficient condition for admissibility within a family of objective functions based on aggregate intercluster similarity. We refer to this family as \emph{sum-type} objective functions. However, apart from Dasgupta's original objective function, no explicit admissible objective functions in this family were provided.
\par
In this paper, we study admissible objective functions for hierarchical clustering in two directions. For sum-type objective functions, we give a complete characterization when the scaling function is a symmetric polynomial of degree at most two, and we derive sufficient conditions for degree-three polynomials. We also show that the recursive sparsest cut 
algorithm achieves an $\order(\phi)$-approximation ratio for the admissible objective functions covered by our characterization, where $\phi$ is the approximation factor of the sparsest cut subroutine. We then introduce \emph{max-type} objective functions, where cluster interaction is measured by maximum, rather than aggregate, intercluster similarity. 
For this class, we characterize which objective functions are admissible for arbitrary symmetric scaling functions and give a complete characterization when the scaling function is a symmetric polynomial of degree at most two.
\end{abstract}
\maketitle
%%%%%%%%%%%%%%%%%%%%%%%%%%%%%%%%%%%%%%%%%%%%%%%%%%%%%%%%%%%%%%%%%%%%%%%%%%%%%%%%%%%%%
%%% Introduction
%%%%%%%%%%%%%%%%%%%%%%%%%%%%%%%%%%%%%%%%%%%%%%%%%%%%%%%%%%%%%%%%%%%%%%%%%%%%%%%%
\section{Introduction} \label{chapter:intro}
Hierarchical clustering is a fundamental procedure in data analysis. Given a finite data set $X$ and a similarity matrix $M$ on $X$, the goal is to represent the data by a nested family of clusters, usually encoded by a rooted binary tree whose leaves are the elements of $X$. Here, a similarity matrix is a nonnegative real-valued symmetric function on $X\times X$. Such a tree is called a \emph{cluster tree} on $X$.
\par
Despite their popularity, classical hierarchical clustering methods have long lacked a principled objective function. Dasgupta~\cite{Dasgupta16} took an important step in this direction by proposing a well-motivated objective function for cluster trees. Let $T=(V,E)$ be a cluster tree on $X$. We denote by $V^\circ$ the set of internal nodes of $T$. For each $v\in V$, let $T_v$ be the subtree rooted at $v$, and let $L(T_v)\subseteq X$ denote the set of leaves of $T_v$. For each internal node $v\in V^\circ$, let $v_+$ and $v_-$ be the two children of $v$. Dasgupta's objective function is defined by
\begin{align}\label{Gamma:Dasgupta}
\Gamma(T) = \sum_{v \in V^{\circ}} H_T(v_+, v_-)(|L(T_{v_+})|+|L(T_{v_-})|),
\end{align}
where
\begin{align*}
H_T(v,w)=\sum_{x\in L(T_v),\,y\in L(T_w)}M(x,y)
\end{align*}
denotes the total similarity between the two clusters.
Dasgupta~\cite{Dasgupta16} showed that minimizing this objective function is NP-hard and gave an $\order(\phi \log n)$-approximation algorithm based on recursive sparsest cut, where $\phi$ is the approximation factor of the sparsest cut subroutine. This approximation ratio was later improved to $\order(\phi)$ by Charikar et al.~\cite{CC} and Cohen-Addad et al.~\cite{CKMM}.
\par
Cohen-Addad et al.~\cite{CKMM} subsequently introduced the notion of \emph{admissible} objective functions through the concept of generating trees. A cluster tree $T$ on $X$ is called a \emph{generating tree} of $M$ if there exists a weight function $h\colon V^\circ\to \mathbb{R}_+$ such that
\begin{enumerate}[{\rm (i)}]
\item $h(u)\leq h(v)$ if $u$ is the parent of $v$, and\label{cond1:generatingTree}
\item $h(\lca_T(x,y))=M(x,y)$ for all distinct $x,y\in X$,\label{cond2:generatingTree}
\end{enumerate}
where $\lca_T(x,y)$ denotes the lowest common ancestor of $x$ and $y$ in $T$.
An objective function is called \emph{admissible} if, whenever the similarity matrix $M$ admits generating trees, the minimizers of the objective function are precisely the generating trees of $M$.
\par
We denote by $\natu$ the set of positive integers. Cohen-Addad et al.~\cite{CKMM} characterized admissibility for a family of objective functions of the form
\begin{align}\label{Gamma:sum}
\Gamma(T)=\sum_{v\in V^\circ} H_T(v_+,v_-)\cdot g(|L(T_{v_+})|,|L(T_{v_-})|),
\end{align}
where $g\colon\natu\times\natu\to\mathbb{R}$ is a symmetric function. In this paper, we refer to objective functions of the form~\eqref{Gamma:sum} as \emph{sum-type} objective functions. Dasgupta's objective function corresponds to the choice $g(a,b)=a+b$. Although the characterization of Cohen-Addad et al. applies to this general family of objective functions, apart from Dasgupta's original objective function, no explicit admissible objective functions in this family were provided. This motivates a more detailed study of admissibility for natural subclasses of scaling functions, such as symmetric polynomials of low degree. Moreover, since sum-type objective functions are based on the aggregate similarity $H_T$ between two clusters, it is natural to ask whether analogous admissibility characterizations can be obtained for objective functions based on other notions of intercluster similarity.
\par
In this paper, we first study sum-type objective functions whose scaling function is a symmetric polynomial. We give a necessary and sufficient condition for admissibility when the scaling function $g$ has degree at most two, and we derive sufficient conditions in the degree-three case. We also show that the recursive sparsest cut algorithm achieves an $\order(\phi)$-approximation ratio for the admissible objective functions covered by our characterization.
\par
We then introduce a new class of objective functions, which we call \emph{max-type} objective functions. A max-type objective function is defined by
\begin{align}\label{Gamma:max}
\Gamma(T) = \sum_{v \in V^{\circ}} K_T(v_+, v_-) \cdot g(|L(T_{v_+})|, |L(T_{v_-})|),
\end{align}
where $g\colon\natu\times\natu\to\mathbb{R}$ is a symmetric function and
\begin{align*}
K_T(v,w)=\max\{M(x,y)\mid x\in L(T_v),\,y\in L(T_w)\}.
\end{align*}
While sum-type objective functions aggregate all pairwise similarities between two clusters through $H_T$, max-type objective functions use only the maximum similarity between the two clusters. For this class, we characterize which objective functions are admissible for arbitrary symmetric scaling functions and give a complete characterization when the scaling function is a symmetric polynomial of degree at most two.
\par
The remainder of this paper is organized as follows. Section 2 studies sum-type objective functions and presents admissibility characterizations for polynomial scaling functions, together with the approximation guarantee for the recursive sparsest cut algorithm. Section 3 introduces max-type objective functions and establishes their admissibility characterizations. Section 4 concludes the paper.
%
%%%%%%%%%%%%%%%%%%%%%%%%%%%%%%%%%%%%%%%%%%%%%%%%%%%%%%%%%%%%%%%%%%%%%%%%%%%%%%%%%%%%%%%%%%%%%%%%%%%%%%%%%%%%%%%%%%%%%%%%%%%%%%%%%%%%%%%
\section{Characterizations of Admissible Sum-Type Objective Functions} \label{chapter:sum}
In this section, we study sum-type objective functions defined in \eqref{Gamma:sum}.
We first characterize admissibility within this class for polynomial scaling functions of low degree.
We then discuss the recursive sparsest cut (RSC) algorithm~\cite{Dasgupta16} and analyze its approximation ratio for the admissible objective functions covered by our characterization.
\subsection{Characterization of Admissible Sum-Type Objective Functions} \label{sec:Chara_sum}
Cohen-Addad et al.~\cite{CKMM} established the following characterization of admissibility for sum-type objective functions for hierarchical clustering.
\begin{proposition}[Cohen-Addad et al.~\cite{CKMM}] \label{prop:admissible_sum}
Let $\Gamma$ be a sum-type objective function, as defined in~\eqref{Gamma:sum}, with scaling function $g\colon\natu\times\natu\to\real$.
Then $\Gamma$ is admissible for every finite set $X$ if and only if it satisfies the following two conditions:
\begin{enumerate}[{\rm(i)}]
\item For every finite set $X$, if $M(x,y)=1$ for all distinct $x,y\in X$ (the uniform similarity case), then $\Gamma(T)$ is constant for all cluster trees $T$ on $X$. \label{cond1:admissible_sum}
\item For any $a,b\in\natu$, $g(a+1,b)>g(a,b)$. \label{cond2:admissible_sum}
\end{enumerate}
\end{proposition}
\par
Building upon Proposition~\ref{prop:admissible_sum}, we study sum-type objective functions whose scaling function $g$ is a symmetric polynomial of degree at most three.
\begin{lemma} \label{lem:Gamma_clique_sum}
Let $\Gamma$ be a sum-type objective function with scaling function $g\colon\natu\times\natu\to\real$. Suppose $g$ is defined by
\begin{align} \label{equ:g=3jiika}
g(a, b) = \lambda((a+b)^3 - (a+b)ab) + \mu(2(a+b)^2 - ab) + \nu(a+b)
\end{align}
for some $\lambda, \mu, \nu \in \real$. 
Then, for every finite set $X$ and every similarity matrix $M$ on $X$ satisfying $M(x, y) = 1$ for all distinct $x, y \in X$, every cluster tree $T$ on $X$ satisfies
\begin{align} \label{equ:Gamma_clique_sum}
\Gamma(T) = \frac{\lambda}{5}(|X|^5 - |X|) + \frac{\mu}{2}(|X|^4 - |X|) + \frac{\nu}{3}(|X|^3 - |X|).
\end{align}
\end{lemma}
\begin{proof}
Let $T$ be an arbitrary cluster tree on $X$ and $n = |X|$. Suppose $M(x, y) = 1$ for all distinct $x, y \in X$. We proceed by induction on $n$.
For $n = 2$, there exists a unique cluster tree $T$ consisting of one root and two leaves. In this case,
\begin{align*}
\Gamma(T) = 1 \cdot g(1, 1) = 6\lambda + 7\mu + 2\nu = \frac{\lambda}{5}(2^5 - 2) + \frac{\mu}{2}(2^4 - 2) + \frac{\nu}{3}(2^3 - 2),
\end{align*}
confirming the base case. Assume the statement holds for $n \leq k - 1$ where $k \geq 3$. For $n = k$, let the root of $T$ be $r$, with $a = |L(T_{r_+})|$ and $b = |L(T_{r_-})|$. By the induction hypothesis:
\begin{align*}
\Gamma(T)
&= H_T(r_+, r_-) \cdot g(a, b) + \Gamma(T_{r_+}) + \Gamma(T_{r_-})\\
&= ab \cdot g(a, b) + \frac{\lambda}{5}(a^5-a+b^5-b) + \frac{\mu}{2}(a^4-a+b^4-b) + \frac{\nu}{3}(a^3-a+b^3-b).
\end{align*}
Substituting \eqref{equ:g=3jiika} and simplifying the terms (using $a+b=n$), we obtain \eqref{equ:Gamma_clique_sum}. Thus, the lemma holds for all $n$.
\end{proof}
%%%%%%%%%%%%%%%%%%%%%%%%%%%%%%%%%%%%%%%%%%%%%%%%%%%%%%%%%%%%%%%%%<<<<<<<<<<<<<<<<<<<<<<<<<<<
\begin{proposition} \label{prop:nsc(i)}
Let $\Gamma$ be a sum-type objective function whose scaling function 
$g(a,b)$ is a symmetric polynomial of degree at most three. Then $\Gamma$ satisfies condition 
\eqref{cond1:admissible_sum} of Proposition \ref{prop:admissible_sum} if and only if $g$ can be expressed in the form \eqref{equ:g=3jiika} for some 
$\lambda,\mu,\nu\in\real$. 
\end{proposition}
\begin{proof}
The sufficiency follows directly from Lemma \ref{lem:Gamma_clique_sum}. To prove necessity, 
let $g(a, b)$ be an arbitrary symmetric polynomial of degree at most three, which can be written in the form 
\begin{align}\label{equ:3.3}
g(a, b) = \lambda_1(a+b)^3 + \lambda_2(a+b)ab + \mu_1(a+b)^2 + \mu_2ab + \nu(a+b).
\end{align}
Since condition~\eqref{cond1:admissible_sum} is assumed to hold for every finite set $X$, we may in particular consider the case $|X|=5$. 
In this case, there are exactly three distinct (non-isomorphic) cluster trees, denoted by $T_1,T_2,$ and $T_3$, 
as shown in Figure~\ref{fig:tree5_max}. Assuming uniform similarity, that is, $M\equiv 1$, 
their objective values are
\begin{align*}
\Gamma(T_1) &= 4g(4, 1) + 3g(3, 1) + 2g(2, 1) + g(1, 1), \\
\Gamma(T_2) &= 4g(4, 1) + 4g(2, 2) + 2g(1, 1), \\
\Gamma(T_3) &= 6g(3, 2) + 2g(2, 1) + 2g(1, 1).
\end{align*}
The condition that $\Gamma(T_1) = \Gamma(T_2) = \Gamma(T_3)$ implies 
\begin{align}
3g(3, 1) + 2g(2, 1) &= 4g(2, 2) + g(1, 1), \label{eq:3241}\\
2g(4, 1) + 2g(2, 2) &= 3g(3, 2) + g(2, 1).\label{eq:2231}
\end{align}
Substituting \eqref{equ:3.3} into \eqref{eq:3241} and \eqref{eq:2231} yields $\lambda_1 + \lambda_2 = 0$ and 
$\mu_1 + 2\mu_2 = 0$, which implies that $g$ 
must take the form~\eqref{equ:g=3jiika}.
\end{proof}
\begin{figure}[th]
\centering
\begin{minipage}[b]{0.3\hsize}
\centering
\includegraphics[scale=0.65]{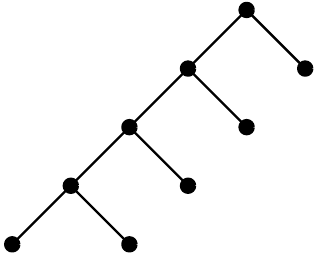}
\subcaption{}
\end{minipage}
\begin{minipage}[b]{0.3\hsize}
\centering
\includegraphics[scale=0.65]{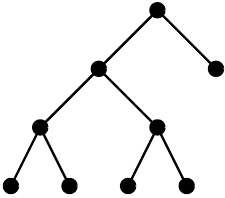}
\subcaption{}
\end{minipage}
\begin{minipage}[b]{0.3\hsize}
\centering
\includegraphics[scale=0.65]{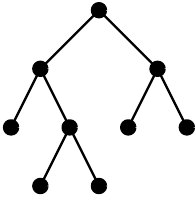}
\subcaption{}
\end{minipage}
\caption{Cluster trees on $X$ with $|X|=5$: (a) $T_1$; (b) $T_2$; (c) $T_3$.}
\label{fig:tree5_max}
\end{figure}
%%%%%%%%%%%%%%%%%%%%%%%%%%%%%%%%%%%%%%%%%%%%%%%%%%%%%%%%%%%%%%%%%%%%%%%%<<<<<<
\par 
Proposition~\ref{prop:nsc(i)} restricts the possible forms of polynomial scaling
functions satisfying the uniform-similarity condition.
We next examine which of these functions satisfy the monotonicity condition,
leading to admissible objectives.
\begin{proposition}\label{theo:3jiika_sum}
Let $g(a,b)$ be a symmetric polynomial of degree at most three. A sufficient condition for a sum-type objective function 
$\Gamma$ to be admissible for every finite set $X$ is that $g$ is of the form \eqref{equ:g=3jiika} with constants $\lambda, \mu, \nu$ satisfying:
\begin{align}\label{equ:3jiika_keisu}
\lambda \geq 0, \quad \mu \geq 0, \quad 15\lambda + 9\mu + \nu > 0.
\end{align}
\end{proposition}
\begin{proof}
Condition \eqref{cond1:admissible_sum} holds by Proposition \ref{prop:nsc(i)}. To verify condition \eqref{cond2:admissible_sum}, 
we consider the difference:
\begin{align*}
g(a+1,b) - g(a,b) = \lambda(3a^2 + 2b^2 + 4ab + 3a + 2b + 1) + \mu(4a + 3b + 2) + \nu.
\end{align*}
Given \eqref{equ:3jiika_keisu}, the minimum of this expression over $a, b \geq 1$ equals $15\lambda + 9\mu + \nu$, 
which is strictly positive. Thus, $\Gamma$ is admissible for every finite set $X$.
\end{proof}
\begin{theorem}
Let $g(a,b)$ be a symmetric polynomial of degree at most two. A necessary and sufficient condition for a sum-type objective 
function $\Gamma$ to be admissible for every finite set $X$ is that $g$ is of the form
\begin{align}\label{equ:g=2jiika}
g(a,b) = \lambda(2(a+b)^2 - ab) + \mu(a+b)
\end{align}
for $\lambda \geq 0$ and $9\lambda + \mu > 0$.
\end{theorem}
\begin{proof}
We first prove sufficiency. The condition \eqref{cond1:admissible_sum} in Proposition \ref{prop:admissible_sum} holds 
directly from Proposition \ref{prop:nsc(i)}. For all $(a, b) \in \natu \times \natu$, we observe that 
\begin{align*}
g(a+1, b) - g(a, b) &= \lambda(4a + 3b + 2) + \mu \\
%&\geq \lambda(4(1) + 3(1) + 2) + \mu\\
& \geq  9\lambda + \mu.
\end{align*}
Given the assumption $9\lambda + \mu > 0$, the condition \eqref{cond2:admissible_sum} is satisfied. 
Thus, $\Gamma$ is admissible by Proposition \ref{prop:admissible_sum}.
\par 
Next, we prove necessity. Assume that $\Gamma$ is admissible for every finite set $X$. By Proposition \ref{prop:admissible_sum}, conditions \eqref{cond1:admissible_sum} 
and \eqref{cond2:admissible_sum} must hold. Since $g$ is a symmetric polynomial of degree at most two, Proposition \ref{prop:nsc(i)} implies that $g$ 
must be in the form of \eqref{equ:g=2jiika} for some $\lambda, \mu \in \real$.
From condition \eqref{cond2:admissible_sum}, we must have
\begin{align}\label{equ:cond2_sum_nec}
g(a+1, b) - g(a, b) = \lambda(4a + 3b + 2) + \mu > 0
\end{align}
for all $a, b \in \natu$. We first show that $\lambda \geq 0$. Suppose, for the sake of contradiction, that $\lambda < 0$. By fixing $b$ and 
taking a sufficiently large $a$, the value of $\lambda(4a + 3b + 2) + \mu$ becomes negative, which contradicts \eqref{equ:cond2_sum_nec}. 
Hence, we must have $\lambda \geq 0$. Furthermore, setting $a = b = 1$ in \eqref{equ:cond2_sum_nec} yields $9\lambda + \mu > 0$, which completes the proof.
\end{proof}
%%%%%%%%%%%%%%%%%%%%%%%%%%%%%%%%%%%%%%%%%%%%%%%%%<<<<<<<
%%%%%%%%%%  Section 2.2  %%%%%%%%%%%%%%%%%%%%%%%%%%%%%%%
%%%%%%%%%%%%%%%%%%%%%%%%%%%%%%%%%%%%%%%%%%%%%%%%%
\subsection{Approximation Algorithm}\label{sec:Approximation_sum}
Having characterized admissible sum-type objective functions in the previous subsection,
we now turn to the algorithmic problem of minimizing such objectives. Namely, 
we consider the following combinatorial optimization problem:
\begin{align}\label{eq:Problem:minGamma}
\begin{array}{ll}
\text{Minimize} & \Gamma(T)\\
\text{subject to} & T \text{ is a cluster tree on } X,
\end{array}
\end{align}
where the objective function $\Gamma$ is a sum-type objective function~\eqref{Gamma:sum} with scaling function $g$ defined by 
\begin{align}\label{equ:g_approximation}
g(a, b) = \lambda((a+b)^3 - (a+b)ab) + \mu(2(a+b)^2 - ab)+ \nu(a+b).
\end{align}
We assume that the coefficients $\lambda, \mu, \nu \in \mathbb{R}$ satisfy
% $\lambda \geq 0, \mu \geq 0,$ and $\lambda + 2\mu + \nu > 0$, 
\begin{align}\label{equ:lmn_appr}
    \lambda \geq 0, \mu \geq 0, \lambda + 2\mu + \nu > 0,
\end{align}
which is a stronger condition than that of Proposition~\ref{theo:3jiika_sum}, and hence ensures 
the admissibility of $\Gamma$.
\par 
To solve Problem~\eqref{eq:Problem:minGamma}, we employ the {\it Recursive Sparsest Cut (RSC) algorithm}. 
Central to this approach is the notion of a {\it  sparsest cut}. For a bipartition $\{Y, X \setminus Y\}$ of $X$, its {\it density} is defined as:
\begin{align*}
d(Y, X \setminus Y) = \frac{\sum_{x \in Y, y \in X \setminus Y} M(x, y)}{|Y|  |X \setminus Y|}.
\end{align*}
Since finding a cut with the minimum density is NP-hard, we employ a subroutine that finds a {\it $\phi$-sparsest cut}---a cut 
whose density is at most $\phi$ times the optimal density (e.g., $\phi = \order(\sqrt{\log n})$ using the ARV algorithm~\cite{ARV}). 
The RSC algorithm (Algorithm~\ref{alg:RSC}) constructs a cluster tree by recursively applying this subroutine.
%%%%%%%%%%%%%%%%%%%%%%%%%%%%%%%%%%%%%%%%%%%%%%%%%%%%%%%%%%%%%%%%%%%%%%%%%%%
\IncMargin{0.5em}
\begin{algorithm}
\caption{Recursive Sparsest Cut Algorithm} \label{alg:RSC}
\KwIn{Similarity matrix $M$ on $X$.}
\If{$|X|=1$}{
\Return the single-leaf tree on $X$\;
}
Find a $\phi$-sparsest cut $\{Y, X \setminus Y\}$ of $M$\;
Recursively obtain cluster trees $T_Y$ and $T_{X \setminus Y}$ for the principal submatrices $M_Y$ and $M_{X \setminus Y}$\;
\Return a tree with a root having $T_Y$ and $T_{X \setminus Y}$ as its children\;
\end{algorithm}
\DecMargin{0.5em}
%%%%%%%%%%%%%%%%%%%%%%%%%%%%%%%%%%%%%%%%%%%%%%%%%%%%%%%%%%%%%%%%%%%%%%%%%%%%
\par 
The analysis follows the charging argument established by Charikar and Chatziafraitis~\cite{CC}. However, the core of our analysis is to 
verify that 
the polynomial $g$ defined in \eqref{equ:g_approximation} satisfies the growth conditions required to obtain an $\order(\phi)$ approximation 
guarantee. 
Before analyzing its performance, we establish the following property regarding the growth of $g$.
\begin{lemma} \label{lem:1}
Consider the function $g(k, n-k)$ for $1 \leq k \leq n-1$, where $g$ is defined in \eqref{equ:g_approximation}. 
This function is minimized when $k \in \{\lfloor n/2 \rfloor, \lceil n/2 \rceil\}$ and maximized when $k \in \{1, n-1\}$. 
\end{lemma}
\begin{proof}
By substituting $a=k$ and $b=n-k$ into \eqref{equ:g_approximation}, we can rewrite $g(k, n-k)$ as a quadratic function of $k$:
\begin{align}\label{eq:g(k,n-k)}
g(k, n-k) &= (\lambda n + \mu)\left(k - \frac{n}{2}\right)^2 + \frac{3}{4}\lambda n^3 + \frac{7}{4}\mu n^2 + \nu n.
\end{align}
If $\lambda n+\mu=0$, then $g(k,n-k)=\nu n$, and the assertion of the lemma holds trivially. Otherwise, the function is a convex parabola 
centered at $k = n/2$. Thus, the function attains its minimum at the integers closest to $n/2$ and its maximum at the boundaries of the domain $k \in [1, n-1]$.
\end{proof}
%%%%%%%%%%%%%%%%%%%%%%%%%%%%%%%%%%%%%%%%%%%%%%%%%%%%%%%%%%%%%%%%%%%<<<<<<<
\par 
Next, we introduce a function $f(t)$ that will be used to charge the cost incurred by the algorithm to the optimal value. 
Although the scaling function $g$ was initially defined on $\natu\times\natu$, 
we extend its domain to $\zahl_+\times\zahl_+$ in a natural way, 
where $\zahl_+$ is the set of nonnegative integers, 
using the same polynomial form \eqref{equ:g_approximation}. 
Note that $g(0,0)=0$ under this definition. We then define $f\colon 
\zahl_+\to\real$ by
\begin{align*}
f(t) = g\left(\left\lfloor\frac{t+1}{2}\right\rfloor, \left\lceil\frac{t+1}{2}\right\rceil\right) - g\left(\left\lfloor\frac{t}{2}\right\rfloor, \left\lceil\frac{t}{2}\right\rceil\right).
\end{align*}
By this definition, for any integer $r\geq 1$, the telescoping sum of $f(t)$ satisfies:
\begin{align}\label{equ:f_sum_property}
\sum_{t=0}^{r-1} f(t) = g\left(\left\lfloor\frac{r}{2}\right\rfloor, \left\lceil\frac{r}{2}\right\rceil\right) - g(0, 0) = g\left(\left\lfloor\frac{r}{2}\right\rfloor, \left\lceil\frac{r}{2}\right\rceil\right).
\end{align}
Using this function $f(t)$, we relate the optimal value $\OPT$ to the similarities between 
clusters. Let $T^*=(V,E)$ be an optimal cluster tree, 
and let
\begin{align*}
\mathcal{C}^* = \{\, L(T_v^*) \mid v \in V \,\}
\end{align*}
denote the family of clusters induced by the nodes of $T^*$. For each $t\in \{0,\dots,n-1\}$, let $\OPT(t)$ 
be the set of maximal (with respect to inclusion) clusters in $\mathcal{C}^*$ whose sizes are at most $t$, and let $E_{\OPT}(t)$ be the set of pairs $\{x,y\}$ such that $x\in A$ and $y\in B$ for distinct $A,B\in\OPT(t)$.

For a set $F$ of unordered pairs in $X$, let $M(F)=\sum_{\{x,y\}\in F}M(x,y)$. For disjoint subsets $A,B\subseteq X$, we write $M(A,B)=\sum_{x\in A,\,y\in B}M(x,y)$. When $F$ is a set of pairs and $A\subseteq X$, we write $F\cap A$ for the set of pairs in $F$ whose two endpoints both belong to $A$.
\begin{lemma}\label{claim:1}
\begin{align}\label{eq:boundOPT}
\sum_{t=0}^{n-1} M(E_{\OPT}(t)) \cdot f(t) \leq \OPT.
\end{align}
\end{lemma}
\begin{proof}
Consider a pair $\{x, y\}$, and let $u=\lca_{T^*}(x,y)$.
Let
\begin{align*}
C_{xy}=L(T_u^*), \qquad C_x=L(T_{u_+}^*), \qquad C_y=L(T_{u_-}^*)
\end{align*}
be the clusters induced by $u$ and its two children. 
The pair's contribution to $\OPT$ is $M(x, y) \cdot g(|C_x|, |C_y|)$. In the left-hand side of~\eqref{eq:boundOPT}, its contribution is 
$\sum_{t=0}^{|C_{xy}|-1} M(x, y) f(t) = M(x, y) \cdot g(\lfloor |C_{xy}|/2 \rfloor, \lceil |C_{xy}|/2 \rceil)$. 
By Lemma~\ref{lem:1}, 
the function $g(a,b)$ is minimized
when $|a-b|$ is minimized among pairs with $a+b=|C_{xy}|$.
Since $|C_x|+|C_y|=|C_{xy}|$, we obtain
\[
M(x,y)g\!\left(\left\lfloor\frac{|C_{xy}|}{2}\right\rfloor,
         \left\lceil\frac{|C_{xy}|}{2}\right\rceil\right)
\le M(x,y)g(|C_x|,|C_y|).
\]
Summing this inequality over all pairs $\{x,y\}$ yields the desired result.
\end{proof}
%%%%%%%%%%%%%%%%%%%%%%%%%%%%%%%%%%%%%%%%%%%%%%%%%%%%%%%%%%%%%%%%%%%%%<<<<<<<
\par 
The final step is to show that the objective function value of the RSC algorithm's output, $\Gamma_{\RSC}$, is bounded by the same sum scaled by $\order(\phi)$. 
The following lemma links the algorithm's recursive splits to the optimal structure. 
\begin{lemma}\label{lm:4.4}
\begin{align*}
\Gamma_{\RSC} \leq \order(\phi) \sum_{A: |A| \geq 2} \frac{s(A)}{|A|} \sum_{t=\lfloor |A|/4 \rfloor}^{\lfloor |A|/2 \rfloor - 1} M(E_{\OPT}(t) \cap A) \cdot f(t).
\end{align*}
\end{lemma}
\begin{proof}
Consider a cluster $A$ generated by the RSC algorithm such that $|A|=r\geq 2$. 
The algorithm splits $A$ into a bipartition $P(A)=\{B_1,B_2\}$ using a $\phi$-approximation 
of the sparsest cut. Let $s(A)=\min\{|B_1|,|B_2|\}$. The contribution of this split to the objective 
function $\Gamma_{\RSC}$ is 
\begin{align*}
M(P(A)) \cdot g(s(A), r - s(A)).
\end{align*}
To bound this cost, we first relate the density of the cut $P(A)$ to the sparsest cut value 
$\SC(A)$. By the definition of the $\phi$-sparsest cut, we have
\begin{align} \label{equ:phi_bound_density}
\frac{M(P(A))}{s(A)(r - s(A))} \leq \phi \SC(A).
\end{align}
We bound $\SC(A)$ using the optimal structure. 
\par
Let 
\begin{equation*}
\{A_1, \dots, A_k\} = \{C \cap A \mid C \in \OPT(\lfloor r/2 \rfloor), C \cap A \neq \emptyset\}.
\end{equation*}
Since $\OPT(\lfloor r/2\rfloor)$ is a partition of $X$ consisting of maximal clusters in $\mathcal{C}^*$ of size at 
most $\lfloor r/2\rfloor$, the set $\{A_1, \dots, A_k\}$ forms a partition of $A$. 
Furthermore, for each $i=1,\dots,k$, we have $|A_i|\leq \lfloor r/2\rfloor\leq r/2$. 
Thus, there exists $0<\gamma_i\leq 1/2$ such that $|A_i|=\gamma_i r$, where 
$\sum_{i=1}^k\gamma_i=1$. 
By the definition of $\SC(A)$, we have 
\begin{align*}
\SC(A) \leq \min_{i=1}^k \frac{M(A_i, A \setminus A_i)}{|A_i| |A \setminus A_i|} \leq \frac{\sum_{i=1}^k M(A_i, A \setminus A_i)}{\sum_{i=1}^k \gamma_i (1 - \gamma_i) r^2}.
\end{align*}
Using the property that $\sum\gamma_i(1-\gamma_i)\geq 1/2$ for $0\leq \gamma_i\leq 1/2$, 
and noting that $\sum M(A_i,A\setminus A_i)=2M(E_{\OPT}(\lfloor r/2\rfloor)\cap A)$, 
we obtain
\begin{align}\label{equ:SC_bound}
\SC(A) \leq \frac{4}{r^2} M(E_{\OPT}(\lfloor r/2 \rfloor) \cap A).
\end{align}
Combining \eqref{equ:phi_bound_density} and \eqref{equ:SC_bound}, the split cost is bounded as:
\begin{align} \label{equ:contribution_A}
\lefteqn{M(P(A)) \cdot g(s(A), r - s(A))}\notag\\
&\leq 4\phi \frac{s(A)(r - s(A))}{r^2} M(E_{\OPT}(\lfloor r/2 \rfloor) \cap A) \cdot g(s(A), r - s(A)).
\end{align}
\par 
Now, define
\[
f'(r)=\sum_{t=\lfloor r/4\rfloor}^{\lfloor r/2\rfloor-1} f(t).
\]
By telescoping, this can be rewritten as
\[
f'(r)
=
g\!\left(
\left\lfloor \frac{\lfloor r/2\rfloor}{2}\right\rfloor,
\left\lceil \frac{\lfloor r/2\rfloor}{2}\right\rceil
\right)
-
g\!\left(
\left\lfloor \frac{\lfloor r/4\rfloor}{2}\right\rfloor,
\left\lceil \frac{\lfloor r/4\rfloor}{2}\right\rceil
\right).
\]
Since $M(E_{\OPT}(t)\cap A)$ is monotonically decreasing in $t$ 
(as clusters in $\OPT(t)$ only coarsen as $t$ increases), we have 
\begin{align} \label{equ:charge_A}
\sum_{t=\lfloor r/4 \rfloor}^{\lfloor r/2 \rfloor - 1} M(E_{\OPT}(t) \cap A) \cdot f(t) \geq M(E_{\OPT}(\lfloor r/2 \rfloor) \cap A) \cdot f'(r).
\end{align}
Substituting \eqref{equ:charge_A} into \eqref{equ:contribution_A}, the split cost of $A$ is
\begin{align*}
\lefteqn{M(P(A)) \cdot g(s(A), r - s(A))}\\
& \leq 4\phi \frac{s(A)}{r} \cdot \frac{g(s(A), r - s(A))}{f'(r)} \sum_{t=\lfloor r/4 \rfloor}^{\lfloor r/2 \rfloor - 1} M(E_{\OPT}(t) \cap A) \cdot f(t).
\end{align*}
As established in Lemma~\ref{lem:c_g} in the Appendix, the ratio 
$\frac{g(s(A),r-s(A))}{f'(r)}$ is bounded from above for our polynomial $g$. Summing over all clusters $A$ generated by the algorithm, we conclude
\begin{align*}
\Gamma_{\RSC} &= \sum_{A} M(P(A)) \cdot g(s(A), |A|-s(A))\\
&\leq \order(\phi) \sum_{A} \frac{s(A)}{|A|} \sum_{t=\lfloor |A|/4 \rfloor}^{\lfloor |A|/2 \rfloor - 1} M(E_{\OPT}(t) \cap A) \cdot f(t).
\end{align*}
\end{proof}
%%%%%%%%%%%%%%%%%%%%%%%%%%%%%%%%%%%%%%%%%%%%%%%%%%%%%%%%%%%%%%%%<<<<<<<
\begin{lemma}\label{claim:2}
\begin{align}\label{eq:lemma2.9}
\sum_{A:|A| \geq 2} \frac{s(A)}{|A|} \sum_{t=\lfloor|A|/4\rfloor}^{\lfloor|A|/2\rfloor-1} M(E_{\OPT}(t) \cap A) \cdot f(t) \leq 2 \sum_{t=0}^{n-1} M(E_{\OPT}(t)) \cdot f(t).
\end{align}
\end{lemma}
\begin{proof}
By swapping the order of summation, we can rewrite the left-hand side of~\eqref{eq:lemma2.9} as 
\begin{align}\label{equ:4.5_new}
\sum_{t=0}^{\lfloor n/2\rfloor-1} f(t) \sum_{A: 2t+2 \leq |A| \leq 4t+3} \frac{s(A)}{|A|} M(E_{\OPT}(t) \cap A).
\end{align}
Here, the inner sum is taken over all clusters $A$ in the algorithm's output such that $t$ falls 
within the 
range $[\lfloor |A|/4\rfloor,\lfloor |A|/2\rfloor -1]$, which is equivalent to $2t+2\leq |A|\leq 4t+3$. 
\par 
For a fixed $t$, we evaluate the inner sum by considering the contribution of each pair 
$\{x,y\}\in E_{\OPT}(t)$. A pair $\{x,y\}$ contributes to the sum only if the cluster $A$ 
contains both $x$ and $y$. Let $A_1\supset A_2\supset \cdots \supset A_k$ be the 
sequence of clusters in the algorithm's output that contain $\{x,y\}$ and 
satisfy the size constraint $2t+2\leq |A_i|\leq 4t+3$. Then,  
\begin{align} \label{equ:Mkeisu_new}
\sum_{A: 2t+2 \leq |A| \leq 4t+3} \frac{s(A)}{|A|} M(E_{\OPT}(t) \cap A) = \sum_{\{x, y\} \in E_{\OPT}(t)} M(x, y) \sum_{i=1}^k \frac{s(A_i)}{|A_i|}.
\end{align}
Recall that $s(A_i)=\min\{|A_{i+1}|,|A_i|-|A_{i+1}|\}$, which implies $s(A_i)\leq |A_i|-|A_{i+1}|$. 
Thus, the inner sum over $i$ is a telescoping-like sum:
\begin{align*}
\sum_{i=1}^k \frac{s(A_i)}{|A_i|} \leq \sum_{i=1}^k \frac{|A_i| - |A_{i+1}|}{|A_i|} \leq \frac{\sum_{i=1}^k (|A_i| - |A_{i+1}|)}{\min_i |A_i|} \leq \frac{|A_1| - |A_{k+1}|}{2t+2},
\end{align*}
where $A_{k+1}$ is a child cluster of $A_k$ that either does not contain $\{x,y\}$
or has size smaller than $2t+2$.
Given the constraint $|A_1|\leq 4t+3$, we have 
\begin{align*}
\frac{|A_1|}{2t+2} \leq \frac{4t+3}{2t+2} < 2.
\end{align*}
Substituting this back into \eqref{equ:Mkeisu_new}, we find that the inner sum is bounded by
$2M(E_{\OPT}(t))$. Finally, applying this to \eqref{equ:4.5_new} yields the desired bound.
\end{proof}
%%%%%%%%%%%%%%%%%%%%%%%%%%%%%%%%%%%%%%%%%%%%%%%%%%%%%%%%%%%%%%%%<<<<<<<
\begin{theorem}
Algorithm \ref{alg:RSC} achieves an $\order(\phi)$-approximation for Problem \eqref{eq:Problem:minGamma}.
\end{theorem}
\begin{proof}
Combining the results of Lemmas \ref{lm:4.4}, \ref{claim:2}, and \ref{claim:1}, we obtain
\begin{align*}
\Gamma_{\RSC} &\leq \order(\phi) \sum_A \frac{s(A)}{|A|} \sum_{t=\lfloor|A|/4\rfloor}^{\lfloor|A|/2\rfloor-1} M(E_{\OPT}(t) \cap A) \cdot f(t)\\
&\leq \order(\phi) \sum_{t=0}^{n-1} M(E_{\OPT}(t)) \cdot f(t)\\
&\leq \order(\phi) \cdot \OPT.
\end{align*}
\end{proof}
%%%%%%%%%%%%%%%%%%%%%%%%%%%%%%%%%%%  Section 3  %%%%%%%%%%%%%%%%%%%%%%%%%%%%%%%%%%%%%%%%%%%%%%%%%
\section{Characterizations of Admissible Max-Type Objective Functions}\label{chapter:max}
In this section, we consider max-type objective functions, as defined in~\eqref{Gamma:max}, 
and characterize admissibility within this class. 
Unlike the sum-type case, we do not address approximation guarantees here.
%%%%%%%%%%%%%%%%%%%%%%%%%%%%%%%%%%%%%%%%%%%%%%%%%%%%%%%%%%%%%%%%%%%%
% Subsection 3.1: Characterization of Max-Type Admissible Functions
%%%%%%%%%%%%%%%%%%%%%%%%%%%%%%%%%%%%%%%%%%%%%%%%%%%%%%%%%%%%%%%%%%%%
\subsection{Characterization of Admissible Max-Type Objective Functions} \label{sec:Chara_max}
The following theorem gives a general characterization of admissible max-type objective functions, with no specific functional restrictions on $g$. 
\begin{theorem}\label{theo:admissible_max}
Let $\Gamma$ be a max-type objective function as defined in~\eqref{Gamma:max}, with scaling function $g\colon \natu\times\natu\to\real$.
Extend $g$ to $\zahl_+\times\zahl_+$ by setting
$g(0,t)=g(t,0)=0$ for all $t\in\zahl_+$.
Then $\Gamma$ is admissible for every finite set $X$ if and only if it satisfies the following conditions:
\begin{enumerate}[{\rm(i)}]
\item For every finite set $X$, if $M(x,y)=1$ for all distinct $x,y\in X$, then $\Gamma(T)$ is equal for all cluster trees $T$ on $X$. \label{cond1:admissible_max}
\item For each $a,b,c,d\in\zahl_+$ such that \label{cond2:admissible_max}
\begin{align*}
a+b>0,\quad c+d>0,\quad a+c>0,\quad b+d>0,\quad a+d>0,\quad b+c>0,
\end{align*}
we have
\begin{align*}
g(a+c,b+d)>g(a,b)+g(c,d).
\end{align*}
\end{enumerate}
\end{theorem}
Before proving Theorem~\ref{theo:admissible_max}, we present the following lemma, which establishes a functional identity for $g$ that is a necessary consequence of condition~\eqref{cond1:admissible_max}.
\begin{lemma}\label{lm:con1}
Assume that condition~\eqref{cond1:admissible_max} holds for every finite set $X$.
Extend $g$ to $\zahl_+\times\zahl_+$ by setting
\[
g(0,t)=g(t,0)=0 \qquad (t\in\zahl_+).
\]
Then, for all $a,b,c,d\in\zahl_+$, we have
\begin{align}\label{con1:equ6_nonnegative}
g(a+b,c+d)+g(a,b)+g(c,d)
=
g(a+c,b+d)+g(a,c)+g(b,d).
\end{align}
\end{lemma}
\begin{proof}
For each positive integer $k$, let $\Gamma^*(k)$ denote the common objective value of cluster trees on a $k$-element set under the uniform similarity matrix. We set $\Gamma^*(0)=\Gamma^*(1)=0$. 
For $p,q\in\zahl_+$, the convention $g(0,t)=g(t,0)=0$ gives
\begin{align*}
\Gamma^*(p+q)=g(p,q)+\Gamma^*(p)+\Gamma^*(q).
\end{align*}
Indeed, when $p,q>0$, this follows by considering a cluster tree whose root separates subsets of sizes $p$ and $q$; when one of $p,q$ is zero, the identity is immediate from the convention.
\par
Applying this identity in two ways to $a+b+c+d$, we obtain
\begin{align*}
\Gamma^*(a+b+c+d)
&= g(a+b,c+d)+\Gamma^*(a+b)+\Gamma^*(c+d)\\
&= g(a+b,c+d)+g(a,b)+g(c,d)\\
&\ +\Gamma^*(a)+\Gamma^*(b)+\Gamma^*(c)+\Gamma^*(d),
\end{align*}
and similarly,
\begin{align*}
\Gamma^*(a+b+c+d)
&= g(a+c,b+d)+\Gamma^*(a+c)+\Gamma^*(b+d)\\
&= g(a+c,b+d)+g(a,c)+g(b,d)\\
&\ +\Gamma^*(a)+\Gamma^*(b)+\Gamma^*(c)+\Gamma^*(d).
\end{align*}
Comparing the two expressions yields~\eqref{con1:equ6_nonnegative}.
\end{proof}

%%%%%%%%%%%%%%%%%%%%%%%%%%%%%%%%%%%%%%%%%%%%%%%%%%%%%%%%%%<<<<<<
\begin{figure}[t]
\centering
\begin{minipage}[b]{0.4\hsize}
\centering
\includegraphics[scale=0.7]{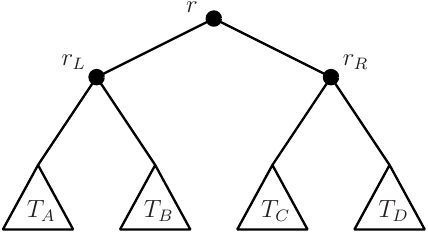}
\subcaption{}\label{con2:T}
\end{minipage}
\hspace{8mm}
\begin{minipage}[b]{0.4\hsize}
\centering
\includegraphics[scale=0.7]{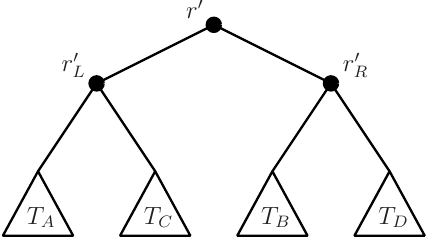}
\subcaption{}\label{con2:T'}
\end{minipage}
\caption{(a) Tree $T$; (b) Tree $T'$.}\label{con2:TT'}
\end{figure}
%%%%%%%%%%%%%%%%%%%%%%%%%%%%%%%%%%%%%
\begin{proof}[Proof of Theorem \ref{theo:admissible_max}]
[The ``only if'' part:] Suppose $\Gamma$ is admissible for every finite set $X$. We show conditions (i) and (ii) hold.
\par
First, to show~\eqref{cond1:admissible_max}, assume $M(x, y) = 1$ for all distinct $x, y \in X$. Any cluster tree $T$ on $X$ is a generating tree of $M$, as a weight function assigning $1$ to all internal nodes satisfies the required conditions. By admissibility, all cluster trees are optimal, thus their objective values are equal.
\par
Next, to show~\eqref{cond2:admissible_max}, let $a, b, c, d \in \zahl_+$ be such that $a+b>0$, $c+d>0$, $a+c>0$, $b+d>0$, $a+d>0$, and $b+c>0$. 
First, we assume $a,b,c,d>0$.
Let $A, B, C, D$ be disjoint subsets with sizes $a, b, c, d$, and $X = A \cup B \cup C \cup D$. Define a similarity matrix $M$ on $X$ as
\begin{align*}
M(x,y) = \begin{cases}
                2 & \text{if } \{x,y\} \subseteq A, B, C, \text{ or } D, \\
                1 & \text{if } (x,y) \in (A \times B) \cup (C \times D), \\
                0 & \text{if } (x,y) \in (A \cup B) \times (C \cup D).
            \end{cases}
\end{align*}
Let $T_A, T_B, T_C, T_D$ be cluster trees on $A, B, C, D$, respectively. Let $T$ be a tree where the root $r$ has children $r_L$ and $r_R$, $r_L$ is the parent of the roots of $T_A$ and $T_B$, and $r_R$ is that of $T_C$ and $T_D$ (Figure~\ref{con2:T}). Similarly, let $T'$ be a tree where the root $r'$ has children $r'_L$ and $r'_R$, $r'_L$ is the parent of the roots of $T_A$ and $T_C$, and $r'_R$ is that of $T_B$ and $T_D$ (Figure~\ref{con2:T'}).
\par
$T$ is a generating tree of $M$ because the assignment of weight $2$ to internal nodes of $T_A, T_B, T_C, T_D$, weight $1$ to $r_L, r_R$, and weight $0$ to $r$ satisfies the definition. However, 
$T'$ is not a generating tree of $M$ since there is no valid weight satisfying the conditions of a generating tree. 
By admissibility, $\Gamma(T') > \Gamma(T)$. Direct calculation gives
\begin{align*}
\Gamma(T) &= g(a, b) + g(c, d) + \Gamma(T_A) + \Gamma(T_B) + \Gamma(T_C) + \Gamma(T_D), \\
\Gamma(T') &= g(a+c, b+d) + \Gamma(T_A) + \Gamma(T_B) + \Gamma(T_C) + \Gamma(T_D).
\end{align*}
Thus, $g(a+c, b+d) > g(a, b) + g(c, d)$. 

If exactly one of $a,b,c,d$ is zero, the same construction is used after omitting the corresponding empty set and the corresponding empty subtree. The convention $g(0,t)=g(t,0)=0$ ensures that the same calculation gives the desired inequality.
\par
%%%%%%%%%%%%%%%%%%%%%%%%%%%%%%%%%%%%%%%%%%%%%%%%%%%%%%%%%%%%%%%%%%%%
% Proof of Theorem 3.1: The "if" part
%%%%%%%%%%%%%%%%%%%%%%%%%%%%%%%%%%%%%%%%%%%%%%%%%%%%%%%%%%%%%%%%%%%%
\noindent [The ``if'' part:] 
We prove the ``if'' part of the theorem by induction on $n=|X|$. 
\par
When $n= 2$, there is only one cluster tree. For any similarity matrix $M$ on $X$, this tree is both the unique generating tree and the unique optimal tree for $\Gamma$; thus, $\Gamma$ is admissible. 
Let $n \geq 3$ and assume the ``if'' part of the theorem holds for objective functions defined for cluster trees on any proper subset of $X$. 
Suppose that $\Gamma$ satisfies conditions~\eqref{cond1:admissible_max} and~\eqref{cond2:admissible_max}. Consider a similarity matrix $M$ on $X$ for which a generating tree exists. Let $T^*$ be a generating tree of $M$ (Figure~\ref{ifpart:T*}), and let $T$ be any cluster tree on $X$ (Figure~\ref{ifpart:T}). 
Let $r^*$ and $r$ denote the root of $T^*$ and $T$, respectively.  
Let $r_L^*$ and $r_R^*$ denote the left and right children of $r^*$, and let $T_L^*$ and $T_R^*$ denote the subtrees 
of $T^*$ rooted at $r_L^*$ and $r_R^*$, respectively.  
    Similarly, let $r_L$ and $r_R$ denote the left and right children of $r$, and let $T_L$ and $T_R$ denote the subtrees of $T$ 
rooted at $r_L$ and $r_R$, respectively.   
    Let $X_L^*, X_R^*, X_L, X_R$ denote the leaf sets of $T_L^*, T_R^*, T_L, T_R$, respectively.  
    Define sets $A, B, C, D$ as $A = X_L^* \cap X_L$, $B = X_L^* \cap X_R$, $C = X_R^* \cap X_L$, $D = X_R^* \cap X_R$, and let $a, b, c, d$ be their respective sizes. 
%%%%
We give the argument for the case $a,b,c,d>0$.
If the root bipartition of $T$ coincides with that of $T^*$ up to exchanging the two sides, then the conclusion follows directly from the induction hypothesis applied to the two child subtrees. Thus, we may assume that the two root bipartitions are distinct. Under this assumption, at most one of $A,B,C,D$ is empty. The case in which exactly one of them is empty is obtained from the following argument by omitting the corresponding empty subtree and all quantities involving the empty set; the convention $g(0,t)=g(t,0)=0$ and Lemma~\ref{lm:con1} ensure that the same algebraic identities remain valid.
%%%
%%%%%%%%%%%%%%%%%%%%%%%% figure %%%%%%%%%%%%%%%%%%%%%%%%%%%%%%%%%%
\begin{figure}[t]
\centering
\begin{minipage}[b]{0.2\textwidth}
\centering
\includegraphics[scale=0.7]{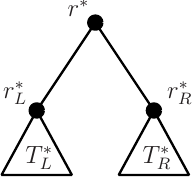}
\subcaption{}\label{ifpart:T*}
\end{minipage}
\hspace{15mm}
\begin{minipage}[b]{0.2\textwidth}
\centering
\includegraphics[scale=0.7]{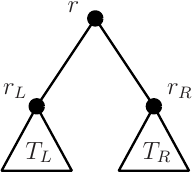}
\subcaption{}\label{ifpart:T}
\end{minipage}
\bigskip\\
\begin{minipage}[b]{0.4\textwidth}
\centering
\includegraphics[scale=0.7]{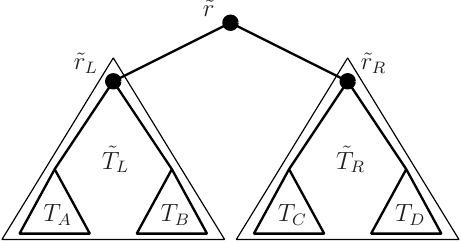}
\subcaption{}\label{ifpart:Ttilde}
\end{minipage}
\hspace{5mm}
\begin{minipage}[b]{0.4\textwidth}
\centering
\includegraphics[scale=0.7]{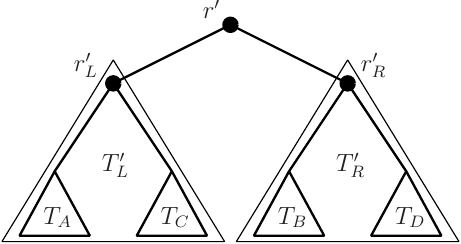}
\subcaption{}\label{ifpart:Tprime}
\end{minipage}
\caption{(a) Tree $T^*$ of $M$; (b) Tree $T$ on $X$; (c) Tree $\tilde{T}$; (d) Tree $T'$.}
\label{ifpart:tree}
\end{figure}
%%%%%%%%%%%%%%%%%%%%%%%% figure %%%%%%%%%%%%%%%%%%%%%%%%%%%%%%%%%%
%%%%%%%%%%%%%%%%%%%%%%%%%%%%%%%%%%%%%%%%%%%%%%%%%%%
\par
For $Y\subseteq X$, let us denote by $M_Y$ the principal submatrix of $M$ with index set $Y$. 
Let $T_A, T_B, T_C, T_D$ be generating trees of the principal submatrices $M_A, M_B, M_C, M_D$, respectively. 
(It follows from Lemma~\ref{lem:GeneratingTree_MyDefe} that any principal submatrix of a similarity matrix that has a generating tree 
also has a generating tree.) By the induction hypothesis, these are optimal trees for their respective principal submatrices. 
Construct a cluster tree $\tilde{T}$ as follows: let $\tilde{r}_L$ be the parent of the roots of $T_A$ and $T_B$, and $\tilde{r}_R$ be the parent of those of $T_C$ and $T_D$. Let $\tilde{r}$ be the root of $\tilde{T}$ with children $\tilde{r}_L$ and $\tilde{r}_R$ (Figure~\ref{ifpart:Ttilde}).
\par
Since $T_L^*$ and $T_R^*$ are generating trees of $M_{A \cup B}$ and $M_{C \cup D}$, respectively, they are respectively optimal for these matrices by the induction hypothesis. Because $L(T_L^*) = L(\tilde{T}_L)$ and $L(T_R^*) = L(\tilde{T}_R)$, it follows that $K_{T^*}(r^*_L, r^*_R) = K_{\tilde{T}}(\tilde{r}_L, \tilde{r}_R)$. Thus, we have
\begin{align}
    \Gamma(T^*) &= K_{T^*}(r^*_L, r^*_R) \cdot g(a+b, c+d) + \Gamma(T_L^*) + \Gamma(T_R^*) \notag \\
    &\leq K_{\tilde{T}}(\tilde{r}_L, \tilde{r}_R) \cdot g(a+b, c+d) + \Gamma(\tilde{T}_L) + \Gamma(\tilde{T}_R) = \Gamma(\tilde{T}). \label{GammaTtilde>=GammaT*}
\end{align}
Suppose that equality holds in~\eqref{GammaTtilde>=GammaT*}. Then $\tilde{T}_L$ and $\tilde{T}_R$ 
are respectively optimal trees for $M_{A\cup B}$ and $M_{C\cup D}$, and hence, by the induction hypothesis, 
generating trees of the respective submatrices, which in turn implies that $\tilde{T}$ is a generating tree of $M$.
\par
Now, construct tree $T'$ (Figure~\ref{ifpart:Tprime}), where the root $r'$ has children $r'_L$ and $r'_R$; $r'_L$ is the parent of the roots of $T_A$ and $T_C$, and $r'_R$ is that of $T_B$ and $T_D$. 
Let $m = \min \{M(x, y) \mid x, y \in X, x \neq y\}$. 
For disjoint subsets $Y,Z$ of $X$, let $m_{YZ} = \max \{M(y, z) \mid y \in Y, z \in Z\}$.  
Since $T^*$ is a generating tree of $M$, $m = m_{AC} = m_{AD} = m_{BC} = m_{BD}$. Let $m' = m_{A \cup C, B \cup D} = \max\{m_{AB}, m_{CD}\}$. Then, 
since $T'_L$ and $T'_R$ are respectively generating trees of $M_{A\cup C}$ and $M_{B\cup D}$, and hence, by the induction 
hypothesis optimal trees for these submatrices, we have 
\begin{align*}
    \Gamma(T) &= m'g(a+c, b+d) + \Gamma(T_L) + \Gamma(T_R) \\
    &\geq m'g(a+c, b+d) + \Gamma(T'_L) + \Gamma(T'_R) \\
    & = \Gamma(T').
\end{align*}
Combining these, we analyze the difference $\Gamma(T) - \Gamma(\tilde{T})$:
\begin{align}
\lefteqn{\Gamma(T) - \Gamma(\tilde{T})}\notag\\
 &\geq \Gamma(T') - \Gamma(\tilde{T}) \notag \\
    &= m'g(a+c, b+d) + m g(a, c) + m g(b, d) \notag \\
    &\quad - mg(a+b, c+d) - m_{AB}g(a, b) - m_{CD} g(c, d) \notag \\
    &= (m' - m)g(a+c, b+d) - (m_{AB} - m)g(a, b)- (m_{CD} - m)g(c, d)  \notag\\
    &\quad + m(g(a+c, b+d) + g(a, c) + g(b, d)) - m(g(a+b, c+d) + g(a, b) + g(c, d)) \notag \\
    &\geq (m' - m)(g(a+c, b+d) - g(a, b) - g(c, d)) \geq 0. \label{GammaT>=GammaTtilde}
\end{align}
The second inequality follows from the identity in Lemma~\ref{lm:con1}, which cancels 
the linear $m$-terms, together with the inequalities $m' \geq m_{AB}$ and $m' \geq m_{CD}$. The third inequality follows from Condition~\eqref{cond2:admissible_max}. 
Note that condition~\eqref{cond2:admissible_max} and the convention
$g(0,t)=g(t,0)=0$ imply $g(p,q)>0$ for all $p,q\in\natu$. 
\par
Since $T$ is an arbitrary cluster tree, \eqref{GammaTtilde>=GammaT*} and \eqref{GammaT>=GammaTtilde} imply $\Gamma(T) \geq \Gamma(T^*)$, meaning any generating tree $T^*$ is optimal.
\par
%%%%%%%%%%%%%%%%%%%%%%%%%%%%%%%%%%%%%%%%%%%%%%%%%%%%%%%%%%%%%%%%%%%%
% Final Segment of Theorem 3.1: Equality, Converse, and Weight Function
%%%%%%%%%%%%%%%%%%%%%%%%%%%%%%%%%%%%%%%%%%%%%%%%%%%%%%%%%%%%%%%%%%%%
\par
We now show the converse: that every optimal tree for $\Gamma$ is a generating tree of $M$. 
Suppose that $T$ is an optimal tree for $M$. Then, from~\eqref{GammaTtilde>=GammaT*} and~\eqref{GammaT>=GammaTtilde}, we must have the equalities $\Gamma(T) = \Gamma(\tilde{T}) = \Gamma(T^*)$. These equalities imply the following:
\begin{enumerate}[{\rm (a)}]
    \item $\tilde{T}$ is a generating tree of $M$,
    \item $T_L$ and $T_R$ are optimal trees for $M_{A\cup C}$ and $M_{B\cup D}$, respectively, and $m_{AB} = m_{CD} = m' = m$.
\end{enumerate}
\par
By (a), there exist constants $\gamma_{AB}$ and $\gamma_{CD}$ such that $M(x,y) = \gamma_{AB}$ for $(x,y) \in A \times B$ and $M(x,y) = \gamma_{CD}$ for $(x,y) \in C \times D$. However, from (b), we must have $m = \gamma_{AB} = \gamma_{CD}$, and hence, 
\begin{align} \label{eq:Ttildesroot}
    M(x,y) = m \quad \text{for all } (x,y) \in (A \times B) \cup (C \times D).
\end{align}
Since $T^*$ is also a generating tree of $M$, we have by definition
\begin{align} \label{eq:T*sroot}
    M(x,y) = m \quad \text{for all } (x,y) \in (A \cup B) \times (C \cup D).
\end{align}
Combining \eqref{eq:Ttildesroot} and \eqref{eq:T*sroot}, we obtain
\begin{align} \label{eq:Tsroot}
    M(x,y) = m \quad \text{for all } (x,y) \in (A \cup C) \times (B \cup D).
\end{align}
By the induction hypothesis and condition (b), $T_L$ and $T_R$ are 
generating trees of $M_{A \cup C}$ and $M_{B \cup D}$, respectively. Thus, there exist weight functions $h_L$ and $h_R$ defined on the internal nodes of $T_L$ and $T_R$ such that:
\begin{align}
   h_L(\lca_T(x,y)) &= M(x,y) \quad (x,y \in A \cup C), \label{eq:h_L} \\
   h_R(\lca_T(x,y)) &= M(x,y) \quad (x,y \in B \cup D). \label{eq:h_R}
\end{align}
Define a weight function $h \colon V^\circ \to \real_+$ for the tree $T$ by:
\begin{align}
h(v) = \begin{cases} 
h_L(v) & \text{if } v \in V^\circ(T_L), \\
h_R(v) & \text{if } v \in V^\circ(T_R), \\
m & \text{if } v = r.
\end{cases}
\end{align}
It follows from \eqref{eq:Tsroot}, \eqref{eq:h_L}, and \eqref{eq:h_R} that $h(\lca_T(x,y)) = M(x,y)$ for all distinct $x,y \in X$. This confirms that $T$ is a generating tree of $M$.
\end{proof}
%%%%%%%%%%%%%%%%%%%%%%%%%%%%%%%%%%%%%%%%%%%%%%%%%%%%%%%%%%%%%%%%%%%%  Subsection 3.2  %%%%%%%%%%%%%%%%%%%%%%%%%%%%%%%%%%%%%%%%%%%%%%%%%%%%%%%%%%%%%%%%%%%%
\subsection{Characterization of Admissible Max-Type Objective Functions with Quadratic Scaling Functions}
 \label{subsec:2jiika}
We give a characterization of admissible max-type objective functions of the form~\eqref{Gamma:max} 
in the case where $g$ is a symmetric polynomial of degree at most two.
\begin{lemma} \label{lem:Gamma_clique_max}
Let $\Gamma$ be a max-type objective function of the form~\eqref{Gamma:max}. Suppose that $g$ is defined by
\begin{align}
        g(a,b) &= \lambda ab\label{equ:g=ab}
\end{align}
for some real number $\lambda$. If $M(x,y) = 1$ for all distinct $x, y \in X$, then for all cluster trees $T$ on $X$, we have 
\begin{align}
   \Gamma(T) = \frac{\lambda}{2}(|X|^2 - |X|). \label{equ:Gamma_clique_max}
\end{align}
\end{lemma}
\begin{proof}
    Let $M$ be a similarity matrix satisfying $M(x,y) = 1$ for all distinct $x, y \in X$, and let $T$ be any cluster tree on $X$. The proof proceeds by induction on $n=|X|$.
    \par
    For $n = 2$, there is only one cluster tree consisting of two leaves and one internal node. For this tree $T$, we have
\begin{align*}
        \Gamma(T) &= g(1,1) = \lambda = \frac{\lambda}{2}(2^2 - 2).
\end{align*}
Thus, the base case holds.
\par
    Assume that \eqref{equ:Gamma_clique_max} holds for $n \leq k-1$ with $k \geq 3$ and 
suppose that $n = k$. Let $r$ be the root of $T$, $a = |L(T_{r_+})|$ and $b = |L(T_{r_-})|$. Then, by the induction hypothesis, 
we have 
\begin{align*}
\Gamma(T) &= g(a,b) + \Gamma(T_{r_+}) + \Gamma(T_{r_-})\\
   &= \lambda ab + \frac{\lambda}{2}(a^2-a) + \frac{\lambda}{2}(b^2-b)\\
%   &= \frac{\lambda}{2}(2ab + a^2 - a + b^2 - b)\\
   &= \frac{\lambda}{2}((a + b)^2 - (a + b))\\
   &= \frac{\lambda}{2}(n^2 - n).
\end{align*}
Thus, the lemma holds for all $n$.
\end{proof}
\begin{theorem} \label{theo:2jiika_max}
Let $g(a,b)$ be a symmetric polynomial of degree at most two.
A necessary and sufficient condition for a max-type objective function $\Gamma$, as defined in~\eqref{Gamma:max}, to be admissible for every finite set $X$ is that
$g$ can be expressed in the form of~\eqref{equ:g=ab} for some real number $\lambda > 0$.
\end{theorem}
\begin{proof}
First, we prove the sufficiency. The condition~\eqref{cond1:admissible_max} in 
Theorem~\ref{theo:admissible_max} holds by Lemma~\ref{lem:Gamma_clique_max}. 
Furthermore, since $\lambda>0$, for any $a,b,c,d\in\zahl_+$ satisfying the conditions in~\eqref{cond2:admissible_max}, we have
\begin{align*}
g(a+c,b+d)-g(a,b)-g(c,d)
&= \lambda\{(a+c)(b+d)-ab-cd\}\\
&= \lambda(ad+bc)>0.
\end{align*}
The last inequality follows from the assumptions
$a+b>0$, $c+d>0$, $a+c>0$, and $b+d>0$.
Thus, the condition~\eqref{cond2:admissible_max} in Theorem~\ref{theo:admissible_max}  also holds. By Theorem~\ref{theo:admissible_max}, $\Gamma$ is admissible.
\par
Next, we prove the necessity. Suppose that $\Gamma$ is admissible for every finite set $X$. By 
Theorem~\ref{theo:admissible_max}, conditions~\eqref{cond1:admissible_max} and \eqref{cond2:admissible_max} hold. Any symmetric polynomial $g(a,b)$ of degree at most two 
can be expressed as 
\begin{align} \label{equ:symmetric}
        g(a,b) = \lambda_1 (a+b)^2 + \lambda_2 ab + \lambda_3 (a+b).
\end{align}
Since admissibility is assumed for every finite set $X$, we may in particular consider the case $|X|=5$. 
In this case, there are only three cluster trees $T_1, T_2$ and $T_3$ as shown 
in Figure~\ref{fig:tree5_max} up to permutations of the leaf set. Assuming $M(x,y) = 1$ for all distinct $x, y \in X$, we have 
\begin{align*}
        \Gamma(T_1) &= g(4,1) + g(3,1) + g(2,1) + g(1,1),\\
        \Gamma(T_2) &= g(4,1) + g(2,2) + 2g(1,1),\\
        \Gamma(T_3) &= g(3,2) + g(2,1) + 2g(1,1).
\end{align*}
By the condition~\eqref{cond1:admissible_max}, $\Gamma(T_1) = \Gamma(T_2) = \Gamma(T_3)$. Thus, we obtain the following equations
\begin{align*}
        g(3,1) + g(2,1) &= g(2,2) + g(1,1),\\
        g(4,1) + g(2,2) &= g(3,2) + g(2,1).
\end{align*}
Substituting equation \eqref{equ:symmetric} into these equations, we obtain
\begin{align*}
        5\lambda_1 + \lambda_3 &= 0,\\
        7\lambda_1 + \lambda_3 &= 0.
\end{align*}
    From these, we deduce $\lambda_1 = \lambda_3 = 0$. Thus, $g$ is of the form
\begin{align*}
        g(a,b) = \lambda ab.
\end{align*}
In addition, $\lambda$ must be positive, since the condition~\eqref{cond2:admissible_max} would be violated otherwise. 
\end{proof}
%%%%%%%%%%%%%%%%%%%%%%%%%%%%%%%%%%%%%%%%%%%%%%%%%%%%%%%%%%%%%%%
\section{Concluding Remarks}

Hierarchical clustering aims to uncover a hierarchical structure of partitions, often represented as a dendrogram or cluster tree, from a data set equipped with pairwise similarities.
Dasgupta~\cite{Dasgupta16} introduced an objective function for evaluating cluster trees and formulated hierarchical clustering as an optimization problem, showing that minimizing this objective is NP-hard.

Cohen-Addad et al.~\cite{CKMM} subsequently introduced the notion of admissible objective functions for hierarchical clustering and gave a necessary and sufficient condition for admissibility within the class of sum-type objective functions~\eqref{Gamma:sum}. 
Although the characterization of Cohen-Addad et al. applies to a general family of objective functions, apart from Dasgupta's original objective function, no explicit admissible objective functions in this family were provided.

In this paper, we studied admissible objective functions for hierarchical clustering in two distinct classes.
For sum-type objective functions, we provided a complete characterization of admissibility when the associated scaling function $g$ is a polynomial of degree at most two, and derived sufficient conditions for admissibility in the degree-three case.
Moreover, for admissible sum-type objective functions, we showed that the recursive sparsest cut algorithm of Dasgupta~\cite{Dasgupta16} achieves an $\order(\phi)$-approximation ratio, where $\phi$ denotes the approximation ratio of the sparsest cut subroutine.

We also introduced a new class of objective functions for hierarchical clustering, termed max-type objective functions.
For this class, we established a general characterization of admissibility without imposing specific functional restrictions on $g$, and further obtained a complete characterization in the case where $g$ is a symmetric polynomial of degree at most two.

These results provide new explicit admissible objective functions and clarify admissibility within both sum-type and max-type classes.
An interesting direction for future work is to compare the theoretical and practical behavior of different admissible objective functions.
Another important open problem is to clarify the computational complexity and approximability of minimizing max-type objective functions. 
%%%%%%%%%%%%%%%%%%%%%%%%%%%%%%%%%%%%%%%%%%%%%%%%%%%%%%%%%%%%%%%%
%%%%%%%%%%%%%%%%%%%%%%%%%%%%%%%%%%%%%%%%%%%%%%%%%%%%%%%%%%%%%%%
%\section*{Acknowledgments}
%%%%%%%%%%%%%%%%%%%%%%%%%%%%%%%%%%%%%%%%%%%%%%%%%%%%%%%%%%%%%%%%
\appendix
\section{The equivalence between the original definition of admissibility and ours} \label{appendixA}
%
% \begin{defe}[Definition]
% 	A function is called admissible if it satisfies the condition: "For any similarity matrix $M$ on $X$ that has a generating tree, the necessary and sufficient condition for a cluster tree $T$ on $X$ to minimize the objective function is that $T$ is the generating tree of $M$."
% \end{defe}
%
The difference between the definition of admissibility in this paper and that of Cohen-Addad et al.~\cite{CKMM} lies in whether the similarity matrix in question is assumed to admit a generating tree or to be generated from an 
ultrametric. Below, we show that these two definitions are equivalent. 
\par 
The following lemma is a similarity-based analogue of the standard representation
theorem for ultrametrics; see, for example, Semple and Steel~\cite{SempleSteel}.
\begin{lemma} \label{lem:GeneratingTree_MyDefe}
    A similarity matrix $M$ on $X$ has a generating tree if and only if for all distinct $x, y, z \in X$,
\begin{align}
M(x,z) \geq \min\{M(x,y), M(y,z)\}. \label{cond:tri_reverse}
\end{align}
\end{lemma}
\begin{proof} 
The argument closely follows the proof of the corresponding representation theorem
for ultrametrics in Semple and Steel~\cite{SempleSteel}, with distances replaced by
similarities and the monotonicity condition reversed accordingly.
\par
[The ``only if'' part:] 
    Assume that $M$ has a generating tree. That is, there exists a cluster tree $T$ on $X$ 
and a weight function $h$ that satisfy the two conditions of a generating tree. 
For any distinct $x, y, z \in X$, since $T$ is a binary tree, two of $\lca_T(x,y)$, $\lca_T(x,z)$, and $\lca_T(y,z)$ are the same internal node, and the other is a descendant of the first two. From the two conditions of a generating tree, if $\lca_T(x,y) = \lca_T(x,z)$, then
    \begin{align*}
        M(x,y) = M(x,z) \leq M(y,z).
    \end{align*}
    Similarly, if $\lca_T(x,y) = \lca_T(y,z)$, then
    \begin{align*}
        M(x,y) = M(y,z) \leq M(x,z).
    \end{align*}
    If $\lca_T(x,z) = \lca_T(y,z)$, then
    \begin{align*}
        M(x,z) = M(y,z) \leq M(x,y).
    \end{align*}
    In all cases, equation \eqref{cond:tri_reverse} holds.
\par
 [The ``if'' part:] We show that if~\eqref{cond:tri_reverse} holds for all distinct $x,y,z\in X$, 
 then $M$ has a generating tree by induction on $|X|$. If $|X| = 3$, let $X = \{x, y, z\}$, and without loss of generality, assume $M(x,z) \geq M(x,y) = M(y,z)$. Let us consider a tree where $x$ and $z$ share a parent $v$, and $v$ and $y$ share a parent $r$. Assign weights $h(v) = M(x,z)$ and $h(r) = M(x,y)$. This tree is a generating tree of $M$.
\par
Let $|X| = n \geq 4$ and assume the ``if'' part of the lemma holds for any 
similarity matrix $M$ on $X$ with $|X| < n$. Let $a, b \in X$ be two distinct 
elements such that $M(a,b)$ is maximized. Then, for any $x \in X \setminus \{a,b\}$, $M(a,b) \geq M(a,x) = M(b,x)$. Let $x_{ab}$ be a new element not in $X$, 
and let $X' = (X \setminus \{a,b\}) \cup \{x_{ab}\}$. Define a similarity matrix $M'\colon X' \times X' \to \real_+$ as follows:
\begin{align*}
M'(x,y) = 
            \begin{cases}
                M(x,y) &\text{if $x,y \in X \setminus \{a,b\}$},\\
                M(a,y) &\text{if $x = x_{ab}$},\\
                M(x,a) &\text{if $y = x_{ab}$}.
            \end{cases}
\end{align*}
 We show that $M'$ satisfies \eqref{cond:tri_reverse} for all distinct $x, y, z \in X'$. If $x, y, z \in X \setminus \{a,b\}$, this is trivial. If $x = x_{ab}$ and $y, z \in X \setminus \{a,b\}$, then
\begin{align*}
        M'(x,z) &= M(a,z)\\
        &\geq \min\{M(a,y), M(y,z)\}\\
        &= \min\{M'(x,y), M'(y,z)\}.
\end{align*}
Therefore, by the induction hypothesis, $M'$ has a generating tree $T'$. Let $h'$ be the weight function for the internal nodes of $T'$ that satisfies the condition for $T'$ to be 
a generating tree of $M'$. Create a new tree $T$ by attaching $a$ and $b$ as children of $x_{ab}$. Assign weights to the internal nodes of $T$ as follows:
\begin{align*}
h(v) =\begin{cases}
                M(a,b) &\text{if $v = x_{ab}$},\\
                h'(v) &\text{if $v \neq x_{ab}$}.
            \end{cases}
\end{align*}
    Let $p$ be the parent of $x_{ab}$. By the definition of $a$ and $b$, $h(p) \leq h(x_{ab})$, so $h$ satisfies the condition~\eqref{cond1:generatingTree}. For any $x \in X \setminus \{a,b\}$, $M(a,x) = M(b,x)$ and $\lca_T(a,x) = \lca_T(b,x)$. To show that $h$ satisfies condition \eqref{cond2:generatingTree}, it suffices to show $h(\lca_T(a,x)) = M(a,x)$. Since $x_{ab}$ is the parent of $a$, we have
    \begin{align*}
        h(\lca_T(a,x)) &= h(\lca_T(x_{ab},x))\\
        &= h'(\lca_T(x_{ab},x))\\
        &= M'(x_{ab},x)\\
        &= M(a,x).
    \end{align*}
\end{proof}
\begin{defe}[Cohen-Addad et al.~\cite{CKMM}]
    A similarity matrix $M$ is called a {\it similarity matrix generated from an ultrametric} if there exists an ultrametric $(X,d)$ and a nonincreasing function $f\colon\real_+ \to \real_+$ such that $M(x,y) = f(d(x,y))$ for all distinct $x, y \in X$.
\end{defe}
\begin{lemma} \label{lem:GeneratingTree_CohenDefe}
    A similarity matrix $M$ on $X$ is generated from an ultrametric if and only if~\eqref{cond:tri_reverse} holds for all distinct $x, y, z \in X$.
\end{lemma}
\begin{proof} 
\noindent[The ``only if'' part:] Suppose that $M$ is generated from an ultrametric. 
Then, there exists an ultrametric $(X,d)$ and a nonincreasing function $f\colon\real_+ \to \real_+$ such that $M(x,y) = f(d(x,y))$ for all distinct $x, y \in X$. Since $d$ satisfies $d(x,z) \leq \max\{d(x,y), d(y,z)\}$ and $f$ is nonincreasing, we have 
\begin{align*}
M(x,z)&=f(d(x,z))\\
 &\geq f(\max\{d(x,y), d(y,z)\})\\
&= \min\{f(d(x,y)), f(d(y,z))\}\\
&=\min\{M(x,y),M(y,z)\}.
\end{align*}
Thus, equation \eqref{cond:tri_reverse} holds.
\par
    [The ``if'' part:] Suppose that $M$ satisfies~\eqref{cond:tri_reverse} 
    for all distinct $x,y,z\in X$.  
Define a constant 
\begin{align*}
c = \max_{x,y \in X, x \neq y} M(x,y) + 1,
\end{align*}
and let $f(\xi) = -\xi + c$. Then, $f$ is nonincreasing. Define $d\colon X\times X\to\real_+$ as follows
\begin{align*}
d(x,y) = \begin{cases}
                -M(x,y) + c &\text{if $x \neq y$},\\
                0 &\text{if $x = y$}.
            \end{cases}
\end{align*}
    We have for all distinct $x, y \in X$, $M(x,y) = -d(x,y) + c = f(d(x,y))$, 
    and for all distinct $x, y, z \in X$,
\begin{align*}
        d(x,z) &= -M(x,z) + c\\
        &\leq -\min\{M(x,y), M(y,z)\} + c\\
        &= \max\{-M(x,y) + c, -M(y,z) + c\}\\
        &= \max\{d(x,y), d(y,z)\}.
\end{align*}
    Thus, $d$ satisfies the strong triangle inequality. By the definitions of $c$ and $d$,  
$d(x,y) \geq 0$ for all 
$x,y\in X$ and $d(x,y) = 0$ if and only if $x = y$. Since $M$ is symmetric, $d(x,y) = d(y,x)$. Therefore, $(X,d)$ is an ultrametric.
\end{proof}
\begin{proposition}
    A similarity matrix $M$ on $X$ has a generating tree if and only if $M$ is generated from an ultrametric.
\end{proposition}
\begin{proof}
    This follows from Lemmas \ref{lem:GeneratingTree_MyDefe} and \ref{lem:GeneratingTree_CohenDefe}.
\end{proof}
%
%
\begin{comment}
\begin{corollary}\label{PrincipalSubmatrix}
    For a similarity matrix $M$ on $X$ that has a generating tree, any principal submatrix $M_A$ induced by a subset $A \subseteq X$ also has a generating tree.
\end{corollary}
%
\begin{proof}
    By Lemma~\ref{lem:GeneratingTree_MyDefe}, $M$ satisfies~\eqref{cond:tri_reverse} 
    for all distinct $x,y,z\in X$, and hence, we have 
\begin{align*}
        M_A(x,z) &= M(x,z)\\
        &\geq \min\{M(x,y), M(y,z)\}\\
        &= \min\{M_A(x,y), M_A(y,z)\}
\end{align*}
for all distinct $x, y, z \in A$. Therefore, $M_A$ also satisfies~\eqref{cond:tri_reverse}. It follows from Lemma~\ref{lem:GeneratingTree_MyDefe} that $M_A$ has a generating tree. 
\end{proof}
\end{comment}
%%%%%%%%%%%%%%%%%%%%%%%%%%%%%%%%%%%%%%%%%%%%%%%%%%%%%%%%%%
\section{Proofs Missing from Section~\ref{sec:Approximation_sum}}\label{appendixB}
%\subsection{Proof of Lemma~\ref{lem:f'}}\label{appendixB1}
%
%
\begin{lemma}\label{lem:f'}
For each integer $t \geq 2$, we have $f'(t) > 0$.
\end{lemma}
\begin{proof}
Using an argument similar to that in the proof of Proposition~\ref{theo:3jiika_sum}, and the coefficient 
condition~\eqref{equ:lmn_appr}, we have 
\begin{align*}%\label{equ:g-g}
g(a+1,b) - g(a,b) &= \lambda(3a^2 + 2b^2 + 4ab + 3a + 2b + 1) + \mu(4a + 3b + 2) + \nu\\
&\geq \lambda + 2\mu+ \nu\\
&>0
\end{align*}
for all $a,b\geq 0$. Thus, $g$ is strictly increasing in its first argument, and by symmetry also in its second argument.
\par
Let $t\geq 2$ and write $t = 4m + r$ for some $m\in\zahl_+$ and $r=0,1,2,3$.
Then, we have 
\begin{align*}
\left\lfloor\left\lfloor\frac{t}{2}\right\rfloor/2\right\rfloor&=m,\\
\left\lceil\left\lfloor\frac{t}{2}\right\rfloor/2\right\rceil&=\begin{cases} m &\text{if $r=0,1$},\\
                                 						m+1 &\text{if $r=2,3$},
						\end{cases}\\
\left\lfloor\left\lfloor\frac{t}{4}\right\rfloor/2\right\rfloor&=\left\lfloor\frac{m}{2}\right\rfloor,\\
\left\lceil\left\lfloor\frac{t}{4}\right\rfloor/2\right\rceil&=\left\lceil\frac{m}{2}\right\rceil,
\end{align*}
and hence, 
\begin{align}\label{equ:r=01_23}
        f'(t) = 
        \begin{cases}
            g\left(m, m\right) - g\left(\left\lfloor\frac{m}{2}\right\rfloor, \left\lceil\frac{m}{2}\right\rceil\right) & \text{if $r = 0, 1$},\\
            g\left(m, m+1\right) - g\left(\left\lfloor\frac{m}{2}\right\rfloor, \left\lceil\frac{m}{2}\right\rceil\right) & \text{if $r = 2, 3$}.
\end{cases}
\end{align}
Furthermore, we have 
\begin{align}
        \left(\left\lfloor\frac{m}{2}\right\rfloor, \left\lceil\frac{m}{2}\right\rceil\right) = \left\{
        \begin{array}{lll}
            \left(\frac{m}{2}, \frac{m}{2}\right) & \text{if} & \text{$m$ is even},\\
            \left(\frac{m-1}{2}, \frac{m+1}{2}\right) & \text{if} & \text{$m$ is odd}.
        \end{array}
        \right. \label{equ:m_is_e_o}
\end{align}
    Using~\eqref{equ:r=01_23} and~\eqref{equ:m_is_e_o}, we have four cases:
\par
    \noindent[Case: $m \geq 2$ is even, $r = 0, 1$]
    \begin{align*}
        f'(t) = g\left(m, m\right) - g\left(\frac{m}{2}, \frac{m}{2}\right).
    \end{align*}
    Since $m > m/2$, $f'(t) > 0$ by the monotonicity of $g$.
    \par
    \noindent[Case: $m \geq 0$ is even, $r = 2, 3$]
    \begin{align*}
        f'(t) = g\left(m, m+1\right) - g\left(\frac{m}{2}, \frac{m}{2}\right).
    \end{align*}
    Since $m \geq m/2$ and $m+1 > m/2$, $f'(t) > 0$ by the monotonicity of $g$.
    \par
    \noindent[Case: $m \geq 1$ is odd, $r = 0, 1$]
    \begin{align*}
        f'(t) = g\left(m, m\right) - g\left(\frac{m-1}{2}, \frac{m+1}{2}\right).
    \end{align*}
    Since $m > \frac{m-1}{2}$ and $m \geq \frac{m+1}{2}$, $f'(t) > 0$ by the monotonicity of $g$.
    \par
    \noindent[Case: $m \geq 1$ is odd, $r = 2, 3$]
    \begin{align*}
        f'(t) = g\left(m, m+1\right) - g\left(\frac{m-1}{2}, \frac{m+1}{2}\right).
    \end{align*}
    Since $m > \frac{m-1}{2}$ and $m+1 > \frac{m+1}{2}$, $f'(t) > 0$ by the monotonicity of $g$.
\end{proof}
%%%%%%%%%%%%%%%%%%%%
%\subsection{Proof of Lemma~\ref{lem:c_g}}\label{appendixB2}
%%%%%%%%%%%%%%%%%%%%%
%%%%%%%%%%%%%%%%%% proof %%%%
%--------------------------------------------------------------------------------------------------------------------------------
\begin{lemma}\label{lem:c_g}
There exists a constant $C$ such that 
\begin{align*}
   \frac{g(s,t-s)}{f'(t)}\leq C
\end{align*}
for all $s,t\in\natu$ such that $t\geq 2$ and $1\leq s\leq t-1$. 
\end{lemma}
\begin{proof}
Since, by Lemmas~\ref{lem:f'} and~\ref{lem:1}, we have
\begin{align*}
 \frac{g(s,t-s)}{f'(t)}\leq  \frac{g(t-1,1)}{f'(t)},
\end{align*}
it suffices to show that $g(t-1,1)/{f'(t)}\ (t\geq 2)$ is bounded from above. 
\par Expressing $t = 4m + r$ as in the proof of  Lemma~\ref{lem:f'}, we classify the cases for $f'(t)$ using~\eqref{equ:r=01_23} and~\eqref{equ:m_is_e_o}.
\par
\noindent[Case: $m \geq 2$ is even, $r = 0, 1$]
\begin{align}
			f'(t) &= g\left(m, m\right) - g\left(\frac{m}{2}, \frac{m}{2}\right) \notag \\
			&= \frac{21}{4}\lambda m^3 + \frac{21}{4}\mu m^2 + \nu m. \label{equ:e01}
\end{align}
\par	
\noindent[Case: $m \geq 0$ is even, $r = 2, 3$]
\begin{align}
			f'(t) &= g\left(m, m+1\right) - g\left(\frac{m}{2}, \frac{m}{2}\right) \notag \\
			&= \frac{21}{4}\lambda m^3 + \left(9\lambda + \frac{21}{4}\mu\right)m^2 + (5\lambda + 7\mu + \nu)m + (\lambda + 2\mu + \nu). \label{equ:e23}
\end{align}
\par	
\noindent[Case: $m \geq 1$ is odd, $r = 0, 1$]
\begin{align}
			f'(t) &= g\left(m, m\right) - g\left(\frac{m-1}{2}, \frac{m+1}{2}\right) \notag \\
			&= \frac{21}{4}\lambda m^3 + \frac{21}{4}\mu m^2 + \left(\nu - \frac{1}{4}\lambda\right)m - \frac{1}{4}\mu. \label{equ:o01}
\end{align}
\par	
\noindent[Case: $m \geq 1$ is odd, $r = 2, 3$]
\begin{align}
			f'(t) &= g\left(m, m+1\right) - g\left(\frac{m-1}{2}, \frac{m+1}{2}\right) \notag \\
			&= \frac{21}{4}\lambda m^3 + \left(9\lambda + \frac{21}{4}\mu\right)m^2 + \left(\frac{19}{4}\lambda + 7\mu + \nu\right)m + \left(\lambda + \frac{7}{4}\mu + \nu\right). \label{equ:o23}
\end{align}
\par
 From \eqref{equ:e01}, \eqref{equ:e23}, \eqref{equ:o01}, and \eqref{equ:o23}, we see that in each of the four cases, 
 $f'(t)$ is a polynomial in $m$ of the form  
\begin{align}\label{eq:f'expr}
			f'(t) = \frac{21}{4}\lambda m^3 + C_1 m^2 + C_2 m + C_3,
\end{align}
where $C_1, C_2, C_3$ are constants independent of $m$ (but may depend on the case). Also, the numerator $g(t-1, 1)$ 
is bounded from above as follows. 
\begin{align}\label{eq:gexpr}
			g(t-1, 1) %&= \lambda(t^3 - t^2 + t) + \mu(2t^2 - t + 1) + \nu t \\
			&= \lambda t^3 + (2\mu - \lambda)t^2 + (\lambda - \mu + \nu)t + \mu \notag\\
			&= 64\lambda m^3 + 48\lambda rm^2 + 12\lambda r^2m + r^3\lambda + 16(2\mu - \lambda)m^2 + 8(2\mu - \lambda)rm \notag\\
			&+ (2\mu - \lambda)r^2 + 4(\lambda - \mu + \nu)m + (\lambda - \mu + \nu)r + \mu \notag\\
			&\leq 64\lambda m^3 + D_1 m^2 + D_2 m + D_3,
\end{align}
where $D_1, D_2, D_3$ are constants independent of $m$ (since $r\in\{0,1,2,3\}$ takes only finitely many values).
\par
If  $\lambda > 0$, then by~\eqref{eq:f'expr} and~\eqref{eq:gexpr}, we have 
\begin{align*}
	g(t-1, 1)/f'(t) &\leq \frac{64\lambda m^3 + D_1 m^2 + D_2 m + D_3}{\frac{21}{4}\lambda m^3 + C_1 m^2 + C_2 m + C_3} \\
	&= \frac{64\lambda + D_1\frac{1}{m} + D_2\frac{1}{m^2} + D_3\frac{1}{m^3}}{\frac{21}{4}\lambda + C_1\frac{1}{m} + C_2\frac{1}{m^2} + C_3\frac{1}{m^3}} \\
	&\rightarrow \frac{256}{21}
\end{align*}
as $m \rightarrow \infty$. Therefore, $g(t-1, 1)/f'(t) $ is bounded from above.
%\par	
The other cases when $\lambda = 0$ and $\mu > 0$ and when $\lambda = \mu = 0$ and $\nu > 0$ 
can  be handled in the same manner. 
\begin{comment}
If $\lambda = 0$ and $\mu > 0$, then we have
\begin{align*}
	g(t-1, 1)/f'(t) &\leq \frac{32\mu m^2 + D_2 m + D_3}{\frac{21}{4}\mu m^2 + C_2 m + C_3} \\
	&= \frac{32\mu + D_2\frac{1}{m} + D_3\frac{1}{m^2}}{\frac{21}{4}\mu + C_2\frac{1}{m} + C_3\frac{1}{m^2}} \\
	&\rightarrow \frac{128}{21} \quad (m \rightarrow \infty).
\end{align*}
Thus, $g(t-1, 1)/f'(t)$ is bounded from above in this case as well.
\par	
For $\lambda = \mu = 0$ and $\nu > 0$:
\begin{align*}
g(t-1, 1)/f'(t)  &\leq \frac{4\nu m + D_3}{\nu m + C_3} \\
	&= \frac{4\nu + D_3\frac{1}{m}}{\nu + C_3\frac{1}{m}} \\
	&\rightarrow 4 \quad (m \rightarrow \infty).
\end{align*}
Thus, $g(t-1, 1)/f'(t)$ is bounded from above in this case as well.
\end{comment}
\end{proof}
\section*{Acknowledgments}
This work was supported by JSPS KAKENHI Grant Numbers 18K11180 and 22K11921.
\section*{Statements and Declarations}
\noindent \textbf{Competing Interests}
The authors declare that they have no competing interests.\\

\noindent \textbf{Declaration of generative AI and AI-assisted technologies in the manuscript preparation process}
During the preparation of this work, the authors used ChatGPT by OpenAI for language editing, proofreading, and improving the clarity and readability of the manuscript. The authors reviewed and edited the output as needed and take full responsibility for the content of the published article. 
%%%%%%%%%%%%%%%%%%%%%%%%%%%%%%%%%%%%%%%%%%%%%%%%%%%%%%%%%%%%%%%%

%%%%%%%%%%%%%%%%%%%%%%%%%%%%%%%%%%%%%%%%%%%%%%%%%%%%%%%%%
\end{document}